\documentclass[11pt]{article}

\usepackage{amssymb}
\usepackage{amsmath}
\usepackage{graphicx}
\usepackage{color}
\usepackage{fullpage}

\begin{document}

\newtheorem{theorem}{Theorem}[section]
\newtheorem{lemma}{Lemma}[section]
\newtheorem{corollary}{Corollary}[section]
\newtheorem{claim}{Claim}[section]
\newtheorem{proposition}{Proposition}[section]
\newtheorem{definition}{Definition}[section]
\newtheorem{fact}{Fact}[section]
\newtheorem{example}{Example}[section]

\newcommand{\afaire}{\dots\\}
\newcommand{\area}{\mathop{Area}}
\newcommand{\diameter}{\mathop{Diameter}}
\newcommand{\algo}{{\tt{LimExpTrav}}}
\newcommand{\algA}{{\tt{LET_A}}}
\newcommand{\algk}{{\tt{LET_k}}}
\newcommand{\alg}{\mathop{LET}}
\newcommand{\explo}{{\tt{ExpCell}}}
\newcommand{\expl}{{\tt{ExpCell}}}
\newcommand{\Null}{\mathop{Null}}
\newcommand{\cell}{\mathop{cell}}
\newcommand{\squared}{\mathop{square}}

\newcommand{\seg}[1]{\overline{#1}}
\def\bbR{\mathbb{R}}
\def\bbN{\mathbb{N}}

\newcommand{\quod}{\hfill $\blacksquare$ \bigbreak}
\newcommand{\reals}{I\!\!R}
\newcommand{\property}{I\!\!P}
\newcommand{\np}{\mbox{{\sc NP}}}
\newcommand{\sing}{\mbox{{\sc Sing}}}
\newcommand{\con}{\mbox{{\sc Con}}}
\newcommand{\prob}{\mbox{Prob}}
\newcommand{\atm}{\mbox{{\sc ATM}}}
\newcommand{\hopn}{\hop_{\cN}}
\newcommand{\atmn}{\atm_{\cN}}
\newcommand{\cA}{{\cal A}}
\newcommand{\cO}{{\cal O}}
\newcommand{\cP}{{\cal P}}
\newcommand{\cC}{{\cal C}}
\newcommand{\cQ}{{\cal Q}}
\newcommand{\C}{{\cal C}}
\newcommand{\cL}{{\cal L}}
\newcommand{\cB}{{\cal B}}
\newcommand{\cG}{{\cal G}}
\newcommand{\cN}{{\cal N}}
\newcommand{\cU}{{\cal U}}
\newcommand{\cF}{{\cal F}}
\newcommand{\cT}{{\cal T}}
\newcommand{\hx}{\hat{x}}
\newcommand{\cS}{{\cal S}}

\newcommand{\EE}{$\epsilon$-\-en\-ve\-lope}
\newcommand{\EQ}{$\epsilon$-\-con\-quer}
\newcommand{\ED}{{\tt En\-ve\-lo\-pe-Dis\-co\-ve\-ry}}
\newcommand{\SC}{{\tt Safe\--Con\-quer}}
\newcommand{\CE}{{\tt Con\-fir\-med-\-E\-cho}}
\newcommand{\CT}{{\tt Co\-lor\&\-Tran\-smit}}
\newcommand{\AC}{{\tt As\-sign-\-Co\-lor}}
\newcommand{\CD}{{\tt Con\-fli\-ct-\-De\-te\-ction}}
\newcommand{\CR}{{\tt Con\-fli\-ct-\-Re\-so\-lu\-tion}}
\newcommand{\SGR}{{\tt Shifted Grid-Refinement}}
\newcommand{\ERRD}{{\tt Error-Detection}}
\newcommand{\WGR}{{\tt Witnessed Grid-Refinement}}

\newcommand{\DB}{$\Delta$-block}
\newcommand{\DBs}{$\Delta$-blocks}
\newcommand{\FDB}{$5\Delta$-block}
\newcommand{\FDBs}{$5\Delta$-blocks}

\newcommand{\UB}{{\tt Universal Broadcast}}
\newcommand{\CAB}{{\tt Company-Aware Broadcast}}
\newcommand{\DI}{{\tt Dense-1}}
\newcommand{\DII}{{\tt Dense-2}}

\newcommand{\MS}{\mathcal{S}}

\newcommand{\eps}{{\epsilon}}
\newcommand{\la}{{\lambda}}
\newcommand{\al}{{\alpha}}
\newcommand{\qed}{\hfill $\square$ \smallbreak}

\newcommand{\UDGI}{{\tt UDG1}}
\newcommand{\UDGII}{{\tt UDG2}}
\newcommand{\SYM}{{\tt SYM}}

\newenvironment{proof}{\noindent{\bf Proof:}}{\qed}

\def\lalto{\left \lceil}
\def\ralto{\right \rceil}
\def\lbasso{\left \lfloor}
\def\rbasso{\right \rfloor}
\def\D{{\Delta}}
\def\qed{\hfill$\Box$}

\baselineskip    0.178in
\parskip         0.05in
\parindent       0.0in

\bibliographystyle{plain}
\title{{\bf Optimal Exploration of Terrains with Obstacles }}
\author{
Jurek Czyzowicz  \footnotemark[1]  \footnotemark[2]
\and
David Ilcinkas \footnotemark[3]
\and
Arnaud Labourel \footnotemark[1] \footnotemark[4]
\and
Andrzej Pelc \footnotemark[1] \footnotemark[5]
}
\date{\today}
\maketitle
\def\thefootnote{\fnsymbol{footnote}}

\footnotetext[1]{ \noindent
D\'epartement d'informatique, Universit\'e du Qu\'ebec en Outaouais, Gatineau,
Qu\'ebec J8X 3X7, Canada.
E-mails: {\tt jurek@uqo.ca}, {\tt labourel.arnaud@gmail.com}, {\tt pelc@uqo.ca}  
}
\footnotetext[2]{ \noindent
Partially supported by NSERC discovery grant.
}
\footnotetext[3]{ \noindent
CNRS, Universit\'{e} de Bordeaux (LaBRI),
33405 Talence, France.
E-mail: {\tt david.ilcinkas@labri.fr}
}
\footnotetext[4]{ \noindent
This work was done during this author's stay at the
Universit\'e du Qu\'{e}bec en Outaouais as a postdoctoral fellow.  
} 
\footnotetext[5]{ \noindent  
Partially supported by NSERC discovery grant and     
by the Research Chair in Distributed Computing at the
Universit\'e du Qu\'{e}bec en Outaouais. 
}

\begin{abstract}
A mobile robot represented by a point moving in the plane has to explore an unknown terrain with obstacles. Both the terrain and the obstacles are modeled as arbitrary polygons. We consider two scenarios: the {\em unlimited vision}, when the robot situated at a point $p$ of the terrain explores (sees) all points $q$ of the terrain for which the segment $pq$
belongs to the terrain, and the {\em limited vision}, when we require additionally that the distance 
between $p$ and $q$ be at most~1. All points of the terrain (except obstacles) have to be explored and the performance of an exploration algorithm is measured by the length of the trajectory of the robot.

For unlimited vision we show an exploration algorithm with complexity $O(P+D\sqrt{k})$, 
where $P$ is the total perimeter of the terrain (including perimeters of obstacles), $D$ is the diameter of the convex hull of the terrain, and $k$ is the number of obstacles.
We do not assume knowledge of these parameters. We also prove a matching lower bound showing that the above complexity is optimal, even if the terrain is known to the robot.
For limited vision we show exploration algorithms with complexity $O(P+A+\sqrt{Ak})$, where $A$ is the area of the terrain (excluding obstacles). Our algorithms work either
for arbitrary terrains, if one of the parameters $A$ or $k$ is known, or for $c$-fat terrains, where $c$ is any constant (unknown to the robot) and no additional knowledge is assumed. (A terrain $\cT$ with obstacles is $c$-fat if $R/r \leq c$, where $R$ is the radius of the smallest disc containing $\cT$ and $r$ is the radius of the largest disc contained in $\cT$.) We also prove
a matching lower bound $\Omega(P+A+\sqrt{Ak})$ on the complexity of exploration for limited vision, even if the terrain is known to the robot.

\vspace*{1cm}

\noindent {\bf keywords:} mobile robot, exploration, polygon, obstacle.

\vspace*{2cm}

\end{abstract}

\thispagestyle{empty}
\setcounter{page}{0}

\pagebreak

\section{Introduction}

{\bf The background and the problem.}
Exploring unknown terrains by mobile robots has important applications when the environment is dangerous or of difficult access for humans. Such is the situation when operating in nuclear plants or cleaning toxic wastes, as well as in the case of underwater or extra-terrestrial operations. In many cases
a robot must inspect an unknown terrain and come back to its starting point. Due to energy and cost saving requirements, the length of the robot's trajectory should be minimized.

We model the exploration problem as follows. The terrain is represented by an arbitrary polygon $\cP _0$ with pairwise disjoint polygonal obstacles $\cP _1,..., \cP _k$, included in $\cP _0$, i.e., the terrain is $\cT=\cP _0 \setminus (\cP _1\cup \cdots \cup \cP _k)$. We assume that borders of all polygons $\cP_i$ belong to the terrain. The robot is modeled as a point moving along a polygonal line inside the terrain. It should be noted that the restriction to polygons is only to simplify the description, and all our results hold in the more general case where polygons are replaced by bounded subsets of the plane homeotopic with a disc (i.e., connected and without holes) and regular enough to have well-defined area and boundary length. Every point of the trajectory of the robot is called \emph{visited}. We consider two scenarios: the {\em unlimited vision}, when the robot visiting a point $p$ of the terrain $\cT$ \emph{explores} (sees) all points $q$ for which the segment $pq$ is entirely contained in $\cT$, and the {\em limited vision}, when we require additionally that the distance between $p$ and $q$ be at most 1. In both cases the task is to explore all points of the terrain $\cT$. The cost of an exploration algorithm is measured by the length of the trajectory of the robot, which should be as small as possible. We assume that the robot does not know the terrain before starting the exploration, but it has unbounded memory and can record the portion of the terrain seen so far and the already visited portion of its trajectory.

{\bf Our results.}
For unlimited vision we show an exploration algorithm with complexity $O(P+D\sqrt{k})$, 
where $P$ is the total perimeter of the terrain (including perimeters of obstacles), $D$ is the diameter of the convex hull of the terrain, and $k$ is the number of obstacles. We do not assume knowledge of these parameters. We also prove a matching lower bound for exploration of some terrains (even if the terrain is known to the robot), showing that the above complexity is worst-case optimal.

For limited vision we show exploration algorithms
with complexity $O(P+A+\sqrt{Ak})$, where $A$ is the area of the terrain\footnote{Since parameters $D,P,A$ are positive reals that may be arbitrarily small, it is important to stress that complexity $O(P+A+\sqrt{Ak})$ means that the trajectory of the robot is at most $c(P+A+\sqrt{Ak})$, for some constant $c$ and
{\em sufficiently large} values of $P$ and $A$. Similarly for $O(P+D\sqrt{k})$. This permits to include, e.g., additive constants in the complexity, in spite of arbitrarily small parameter values.}.
Our algorithms work either for arbitrary terrains, if one of the parameters $A$ or $k$ is known, or for $c$-fat terrains, where $c$ is any constant larger than $1$ (unknown to the robot) and no additional knowledge is assumed. (A terrain $\cT$ is $c$-fat if $R/r \leq c$, where $R$ is the radius of the smallest disc containing $\cT$ and $r$ is the radius of the largest disc contained in $\cT$.) 
We also prove a matching lower bound $\Omega(P+A+\sqrt{Ak})$ on the complexity of exploration, even if the terrain is known to the robot.

The main open problem resulting from our research is whether exploration with asymptotically optimal cost $O(P+A+\sqrt{Ak})$ can be performed in arbitrary terrains without \emph{any} a priori knowledge.

{\bf Related work.}
Exploration of unknown environments by mobile robots was extensively
studied both for the unlimited and for the limited vision. Most of the research in
this domain concerns the competitive framework, where the trajectory of the robot not knowing
the environment is compared to that of the optimal exploration algorithm having full knowledge.  

One of the most
important works for unlimited vision is~\cite{DKP98}. 
The authors gave a $2$-competitive algorithm for
rectilinear polygon exploration without obstacles. 
The case of non-rectilinear polygons (without obstacles) 
was also studied in~\cite{DKP91,HIKK01}
and a competitive algorithm was given in this case.

For polygonal environments with an arbitrary number of polygonal
obstacles, it was shown in~\cite{DKP98} that no competitive strategy
exists, even if all obstacles are parallelograms. Later, this result
was improved in~\cite{AKS02} by giving a lower bound in
$\Omega(\sqrt{k})$ for the competitive ratio of any on-line algorithm
exploring a polygon with $k$ obstacles. This bound remains
true even for rectangular
obstacles. On the other hand, there exists an algorithm with
competitive ratio in $O(k)$~\cite{DKP91}.

Exploration of polygons by a robot with limited vision has been studied,
e.g., in \cite{GB01,GB03,GBBS08,IKRL00,KKMZ09,N92}. In \cite{GB01} 
the authors described an on-line algorithm with
competitive ratio $1+3(\Pi S/A)$, where $\Pi$ is a quantity depending
on the perimeter of the polygon, $S$ is the area seen by the robot, and $A$ is the area of
the polygon. The exploration in~\cite{GB01,GB03} fails on a certain type of polygons, such
as those with narrow corridors. In \cite{GBBS08}, the authors consider exploration in discrete steps.
The robot can only explore the environment when it is motionless, and the cost of the exploration algorithm is measured by the number of stops during the exploration. In \cite{IKRL00,KKMZ09}, the complexity of exploration is measured by the trajectory length, but only terrains composed of identical squares  are considered. In \cite{N92} the author studied off-line exploration of the
boundary of a terrain with limited vision.

An experimental approach was used in \cite{BLAS} to show the performance of a greedy
heuristic for exploration in which the robot always moves to the frontier between explored
and unexplored area.
Practical exploration of the environment by an actual robot was studied, e.g., in \cite{CQG,SM}.
In \cite{SM}, a technique is described to deal with obstacles that are not in the
plane of the sensor. In \cite{CQG} landmarks are used during exploration to construct
the skeleton of the environment. 

Navigation is a closely related task which consists in finding a path between two given points in a terrain with unknown obstacles. Navigation in a $n\times n$ square containing rectangular obstacles
aligned with sides of the square was considered in \cite{BBFY,BPFKRS,BRS,PY}. It was shown in \cite{BBFY} that the navigation from a corner to the center of a room can be performed with a
competitive ratio $O(\log n)$, only using tactile information (i.e., the robot modeled as a point sees an obstacle only when it touches it). No deterministic algorithm can achieve better competitive ratio, even with unlimited vision~\cite{BBFY}. For navigation between any pair of points, there is a deterministic algorithm achieving a competitive ratio of $O(\sqrt{n})$~\cite{BRS}. No deterministic
algorithm can achieve a better competitive ratio~\cite{PY}. However, there is a randomized approach performing navigation with a competitive ratio of $O(n^{\frac{4}{9}}\log
n)$~\cite{BPFKRS}.

Navigation with little information was considered in~\cite{TML}. In this model,
the robot cannot perform localization nor measure any distances or angles.
Nevertheless, the robot is able to learn the critical information contained in the
classical shortest-path roadmap and perform locally optimal navigation.

\section{Unlimited vision}

Let $S$ be a smallest square in which the terrain $\cal T$ is included. Our algorithm constructs a 
\emph{quadtree decomposition} of $S$. A quadtree is a rooted tree with each non-terminal node having four children. Each node of 
the quadtree corresponds to a square. The children of any non-terminal node $v$ correspond to four identical squares obtained by 
partitioning the square of $v$ using its horizontal and vertical symmetry axes. This implies that the squares of the terminal nodes 
form a partition of the root\footnote{In order to have an exact partition we assume that each square of the quadtree partition 
contains its East and South edges but not its West and North edges.}. More precisely,
\begin{enumerate}
\item $\{S\}$ is a quadtree decomposition of $S$
\item If $\{S_1, S_2, \dots , S_j\}$ is a quadtree decomposition of $S$, then\\ 
$\{S_1, S_2, \dots, S_{i-1}, S_{i_1}, S_{i_2}, 
S_{i_3}, S_{i_4}, S_{i+1}, \dots, S_j\}$, where $S_{i_1}, S_{i_2}, S_{i_3}, S_{i_4}$ form a partition of $S_i$ using its vertical and horizontal symmetry axes, is a quadtree decomposition of $S$
\end{enumerate}

The trajectory of the robot exploring $\cal T$ will be composed of parts which will follow the boundaries
of ${\cal P} _i$, for $0 \leq i \leq k$, and of straight-line segments, called \emph{approaching segments},  joining the boundaries of ${\cal P} _i$,  $0 \leq i \leq k$. Obviously, the end points of an approaching segment must be visible from each other. The quadtree decomposition will be dynamically constructed in a top-down manner during the exploration of $\cal T$.  At each moment of the exploration we consider the set $\cal{Q}_S$ of all squares of the current quadtree and the set $\cal{Q}_T$ of squares being the terminal nodes of the current quadtree. We will also construct dynamically a bijection $f:\{{\cal P} _0, {\cal P} _1, \dots, {\cal P} _k\} \longrightarrow \cal{Q_S} \setminus \cal{Q_T}$.

When a robot moves along the boundary of some polygon ${\cal P} _i$, it may be in one of two possible modes: the \emph{recognition mode} - when it goes around the entire boundary of a polygon without any deviation, or in the \emph{exploration mode} - when, while moving around the boundary, it tries to detect (and approach) new obstacles. When the decision to approach a new obstacle is made at 
some point $r$ of the boundary of ${\cal P} _i$ the robot moves along an approaching segment to reach the obstacle, processes it by a recursive call, and (usually much later), returning from the recursive call, it moves again along this segment in the opposite direction in order to return to point $r$ and to continue the exploration of ${\cal P} _i$. However, some newly detected obstacles may not be immediately approached. We say that, when the robot is in position $r$, an obstacle ${\cal P} _j$ is \emph{approachable}, if there exists a point $q \in {\cal P} _j$, 
belonging to a square $S_t \in \cal{Q_T}$ of diameter $D(S_t)$ such that $|rq| \leq 2D(S_t)$.
It is important to state that if exactly one obstacle becomes approachable at moment t, then it is approached immediately and if more than one obstacle become approachable at a moment $t$, 
then one of them (chosen arbitrarily) is approached immediately and the others are approached 
later, possibly from different points of the trajectory. Each time a new obstacle is visited by the robot (i.e., all the points of its boundary are visited in the recognition mode)
the terminal square of the current quadtree containing the first visited point of the new obstacle is partitioned. This square is then associated  to this obstacle by function $f$. 

The trajectory of the robot is composed of three types of sections: \emph{recognition sections}, \emph{exploration sections} and \emph{approaching sections}. The boundary of each polygon will be traversed twice: first time contiguously during a recognition section and second time through exploration sections, which may be interrupted several times in order to approach and visit newly detected obstacles. We say that an obstacle is \emph{completely explored}, if each point on the boundary of this obstacle has been traversed by an exploration section. We will prove that the sum of the lengths of the approaching sections is $O(D \sqrt{k})$.

\begin{tabbing}
{\bf Algorithm} {\tt ExpTrav} (polygon $R$, starting point $r^*$ on the boundary of $R$) \\
1 \= Make a recognition traversal of the boundary of $R$ \\
2 \= Partition square $S_t \in Q_T$ of the quadtree containing $r^*$ into four identical squares\\
3 \= $f(R) := S_t$ \\
4 \> {\bf rep}\={\bf eat}  \\
5 \>  \> Trav\=erse the boundary of $R$ until, for the current position $r$, there exists a visible point $q$\\
\> \> \>of a new obstacle $Q$ belonging to square $S_t \in \cal{Q}_T$, such that $|rq| \leq 2D(S_t)$ \\
6 \> \> Traverse the segment $rq$\\
7 \> \> ExpTrav($Q$, $q$) \\
8 \> \> Traverse the segment $qr$\\
9 \> {\bf until} $R$ is completely explored 
\end{tabbing}

Before the initial call of {\tt ExpTrav}, the robot reaches a position $r_0$ at the boundary of the polygon $\cP_0$. This is done as follows. At its initial position $v$, the robot chooses an arbitrary half-line $\alpha$ which it follows as far as possible. When it hits the boundary of a polygon $\cP$, it 
traverses the entire boundary of $\cP$. Then, it computes the point $u$ which is
the farthest point from $v$ in $\cP\cap \alpha$. It goes around $\cP$ until reaching $u$ again and progresses on $\alpha$, if possible. If this is impossible, the robot recognizes
that it went around the boundary of $\cP_0$ and it is positioned on this boundary. It initialises the quadtree decomposition to a smallest square $S$ containing $\cP_0$. This square is of size $O(D({\cal P} _0))$.
The length of the above walk is less than $3P$.

\begin{lemma}
\label{l:explorability}
Algorithm {\tt ExpTrav} visits all boundary points of all obstacles of the terrain $\cal T$.
\end{lemma}

\begin{proof}
Note that Algorithm {\tt ExpTrav} always 
terminates. Indeed, since there is a finite number of obstacles, there is a finite number of calls of {\tt ExpTrav} and steps 5-8 are executed a finite number of times. Moreover, since each obstacle has a finite boundary, step 1 and the repeat loop are always completed.

Consider the quadtree decomposition ${\cal Q}$ of $S$ arising at the completion of the algorithm. 
Suppose, for contradiction, that there exists a point on the boundary of an obstacle which was never visited. Let $p$ be a point among all unvisited boundary points for which the terminal square $S_j$ belonging to the quadtree $Q$ has the smallest possible diameter. Consider the square $S_m$, the parent of $S_j$ in ${\cal Q}$. $S_m$ was partitioned in step 2 of some call of {\tt ExpTrav} as a result of detecting some obstacle ${\cal P}'$ intersecting $S_m$. Consider the segment $qp$, where $q \in S_m$ belongs to the boundary of ${\cal P}'$. Since both points $p$ and $q$ belong to the boundary of $\cT$ and $q$ was visited while $p$
was not, there exists a pair of points $q'$ and $p'$, both belonging to the boundary of $\cT$ and to the segment $qp$, such that $q'$ was visited, $p'$ was not and $p'$ is visible from $q'$. Such a pair exists because at the end of the exploration the boundary of each polygon is either entirely visited or not at all.
Consider the quadtree at the moment $t$ when the robot visited point $q'$, and its terminal square $S_i$ containing point $p'$. Clearly, $D(S_i)\geq D(S_j)$, because $S_j$ is a square with the smallest 
diameter containing unvisited boundary points. Hence 
$|q'p'|\leq |qp|\leq D(S_m)=2D(S_j)\leq 2D(S_i)$ and $p'$ was approachable from $q'$ at time $t$,
a contradiction.
\end{proof}

\begin{lemma}
\label{l:bijection}
Function $f$ is a bijection from $\{{\cal P} _0, {\cal P} _1, \dots, {\cal P} _k\}$ to $\cal{Q_S} \setminus \cal{Q_T}$,
where $\cal{Q_S}$ and $\cal{Q_T}$ correspond to the final quadtree decomposition produced by Algorithm {\tt ExpTrav}.
\end{lemma}

\begin{proof}
When {\tt ExpTrav} is called for the first time, the robot is on the boundary of ${\cal P} _0$ and the quadtree has exactly one 
non-terminal node - its root $S$, and $f({\cal P} _0)=S$. By induction, each time a new obstacle $Q$ is approached in step 6 of a 
call of {\tt ExpTrav}, $f(Q)$ is set to some $S_t$ intersecting $Q$ and $S_t$ becomes a nonterminal node of the quadtree in step 2. Hence 
each square corresponding to a non-terminal node of the quadtree is an image of a different polygon ${\cal P} _i$,  $0 \leq i \leq 
k$.
\end{proof}

\begin{lemma}
\label{l:quadtree-complexity}
For any quadtree $T$,  rooted at a square of diameter $D$ and having $x$ non-terminal nodes, the sum $\sigma(T)$ of diameters of these nodes is at most $2D\sqrt{x}$.
\end{lemma}

\begin{proof}
The diameter of a square at depth $t$ of the quadtree is $\frac{D}{2^{t}}$. We prove first that among all quadtrees with $x$ 
non-terminal nodes, $\sigma(T)$ is maximized for the quadtree having all terminal nodes of at most two consecutive depths. Suppose, 
to the contrary, that there exists a quadtree $T$ maximizing $\sigma(T)$ having terminal nodes of two different depths $t_1$ and 
$t_2$ with $t_1 < t_2 - 1 $. Let $p$ be a terminal node of $T$ of depth $t_1$ and $q$ be a non-terminal node of depth $t_2-1$. Let $T'$ be  the tree obtained from $T$ by detaching from $T$ node $p \in T$ and the subtree of $T$ rooted at $q$ and 
exchanging their places. $T'$ is again a quadtree with $x$ non-terminal nodes, having one less non-terminal node of depth $t_2-1$ 
and one extra non-terminal node of depth $t_1$. Hence $\sigma(T') \geq \sigma(T) - \frac{D}{2^{t_2-1}}+\frac{D}{2^{t_1}} > \sigma(T)$, 
which contradicts the maximality of $\sigma(T)$.

Therefore it is sufficient to consider only quadtrees having terminal nodes of at most two consecutive depths $t$ and $t+1$. Suppose that there 
are $y$ terminal nodes of depth $t+1$, $0\leq y \leq 4^{t+1}$. Then the number $x$ of non-terminal nodes equals $\frac{4^{t} - 1}{3} + \frac{y}{4}$ and 
\[ \sigma(T)=\frac{y}{4}\frac{D}{2^{t}}+\sum_{i=0}^{t-1} 4^{i}\frac{D}{2^{i}} \]
\[ = D\left( 2^{t}-1 + \frac{y}{2^{t+2}}\right) \]

We need to prove that $\sigma(T) \leq 2D\sqrt{x}$, i.e. that
\[ D\left( 2^t-1+\frac{y}{2^{t+2}}\right) \leq 2D\sqrt{\frac{4^t-1}{3}+\frac{y}{4}}\]
Hence it is sufficient to show that
\[ \left( 2^t-1+\frac{y}{2^{t+2}}\right)^2 \leq 4\left( \frac{4^t-1}{3}+\frac{y}{4}\right)\]
i.e., that
\[0 \leq y\left(1+\frac{1}{2^{t+1}}-\frac{1}{2}-\frac{y}{4^{t+2}}\right)+\left(4\cdot 
\frac{4^t-1}{3}-2^{2t}-1+2^{t+1}\right) \]
The second term is clearly positive for $t \geq 0$ and the first term is also positive since $0\leq y \leq 4^{t+1}$. We conclude that $\sigma(T) \leq 2D\sqrt{x}$.
\end{proof}

\begin{theorem}
\label{t:exploration-correctness}
Algorithm {\tt ExpTrav} explores the terrain $\cal T$ of perimeter $P$ and convex hull diameter $D$ with $k$ obstacles in time $O(P+D\sqrt{k})$.
\end{theorem}

\begin{proof}
Take an arbitrary point $p$ inside  $\cal T$ and a ray outgoing from $p$ in an arbitrary direction. This ray reaches the boundary 
of  $\cal T$ at some point $q$. Since, by Lemma \ref{l:explorability} point $q$ was visited by the robot, $p$ was visible from $q$ during the robot's traversal, and hence $p$ was explored.

To prove the complexity of the algorithm, observe that the robot traverses twice the boundary of each polygon of $\cal T$, once during its recognition in step 1 and the second time during the iterations of step 5. Hence the sum of lengths of the recognition and exploration sections is $2P$. The only other portions of the trajectory are produced in steps 6 and 8, when the obstacles are 
approached and returned from. According to the condition from step 5, an approaching segment is traversed in step 6 only if its length is shorter than twice the diameter 
of the associated square. If $k=0$ then the sum of lengths of all approaching segments is $0$, due to the fact that exploration starts at the external boundary of the terrain. In this case the length of the trajectory is at most $5P$. Hence we may assume that $k>0$. By Lemma \ref{l:bijection} each obstacle is associated with a different non-terminal node of the quadtree and  the number $x$ of non-terminal nodes of the quadtree equals $k+1$. Hence the sum of lengths of all approaching segments is at most $2\sigma(T)$. By Lemma \ref{l:quadtree-complexity} we have $\sigma(T) \leq 2D\sqrt{x}=2D\sqrt{k+1}$,
hence the sum of lengths of approaching segments is at most $2\sigma(T) \leq 4D\sqrt{k+1}\leq 4D\sqrt{2k}\leq 6D\sqrt{k}$. Each segment is traversed twice, so the total length of this part of the trajectory is at most  $12D\sqrt{k}$. It follows that the total length of the trajectory is at most $5P+12D\sqrt{k}$.
\end{proof}
\begin{theorem}
\label{t:lower-unlimited}
Any algorithm for a robot with unlimited visibility, exploring polygonal terrains with $k$ obstacles, having total perimeter $P$ and the convex hull diameter $D$, produces trajectories in $\Omega(P+D\sqrt{k})$ in some terrains, even if the terrain is known to the robot.
\end{theorem}
\begin{proof}
In order to prove the lower bound, we show two families of terrains: one for which $P\in\Theta(D)$
($P$ cannot be smaller), $D$ and $k$ are unbounded and still the exploration cost is $\Omega(D\sqrt{k})$, and the other in which $P$ is unbounded, $D$ is arbitrarily small, $k=0$ and still the exploration cost is $\Omega(P)$.

Consider the terrain from Figure \ref{fig:lower-unlimited}(a) where $k$ identical tiny obstacles are distributed evenly at the $\sqrt{k} \times \sqrt{k}$ grid positions inside a square of diameter $D$. The distance between obstacles is at least $\frac{D\sqrt{2}}{2(\sqrt{k}+1)}-\epsilon$ where 
$\epsilon > 0$ may be as small as necessary by choosing obstacles  sufficiently small. The obstacles are such that to explore the small area inside the convex hull of the obstacle the robot must enter this convex hull. Since each such area must be explored, the trajectory of the robot must be of size at least $(k-1)\left(\frac{D\sqrt{2}}{2(\sqrt{k}+1)}-\epsilon\right)$, which is clearly in $\Omega(D\sqrt{k})$. Note that the perimeter $P$ is in $\Theta(D)$.

The terrain from Figure \ref{fig:lower-unlimited}(b) is a polygon of arbitrarily small diameter (without obstacles), whose exploration requires a trajectory of size $\Omega(P)$, where $P$ is unbounded. Indeed, each "corridor" must be traversed almost completely to explore points at its end. Hence the two families of polygons from Figure \ref{fig:lower-unlimited} lead to the $\Omega(P+D\sqrt{k})$ lower bound.
\end{proof}
\begin{figure}[h]
	\begin{center}
	\includegraphics{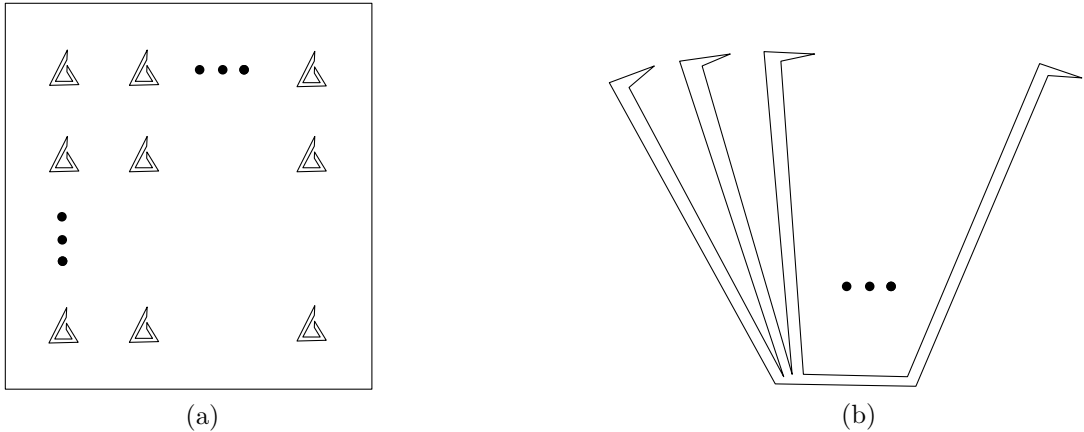}
	\caption{Lower bound for unlimited visiblity}
	\label{fig:lower-unlimited}
	\end{center}
\end{figure}

\section{Limited vision}

In this section we assume that the vision of the robot has range 1.
The following algorithm is at the root of all our positive results on exploration
with limited vision. The idea of the algorithm is to partition the environment
into small parts called \emph{cells} (of diameter at most $1$) and to visit them using a depth-first traversal. 
The local exploration of cells can be performed using Algorithm {\tt ExpTrav},
since the vision inside each cell is not limited by the range $1$ of the vision of 
the robot. The main novelty of our exploration algorithm is that the robot completely
explores \emph{any} terrain. This should be contrasted with previous algorithms with limited 
visibility, e.g. \cite{GB01,GB03,IKRL00,KKMZ09} in which only a particular class of terrains with obstacles is explored, e.g., terrains without narrow corridors or terrains composed of complete identical squares. This can be done at cost $O(A)$. Our lower bound shows that exploration complexity of arbitrary terrains depends on the perimeter and the number of obstacles as well. The complete exploration of arbitrary terrains achieved by our algorithm significantly complicates both the exploration process and its analysis. 

\noindent
{\bf Algorithm $\algo$} ($\alg$, for short)\\
INPUT: A point $s$ inside the terrain $\cT$ and a positive real $F\leq \sqrt{2}/2$.\\
OUTPUT: An exploration trajectory of $\cT$, starting and ending at~$s$.

\noindent
Tile the area with squares of side $F$, such that $s$ is on the boundary of a square. The connected regions obtained as intersections of $\cT$ with each tile
are called \emph{cells}. For each tile $S$, maintain a quadtree decomposition $Q_S$
initially set to $\{S\}$. Then, arbitrarily choose one of the cells containing $s$ to be the starting cell $C$ and call $\explo$($C$, $s$).
\begin{tabbing}
{\bf Procedure} $\explo$(current cell $C$, starting point $r^*\in C$) \\
1 \= Record $C$ as visited\\
2 \= {\tt ExpTrav}($C$,$r^*$) using the quadtree decomposition $Q_S$, 
where $S$ is the tile containing $C$\\
3 \textbf{rep}\=\textbf{eat} \\
4 \> Traverse the boundary of $C$ until the current position $r$ belongs to an unvisited cell $U$\\
5 \> $\explo$($U$, $r$)\\
\> (if $r$ is in several unvisited cells, choose arbitrarily the first cell to be processed)\\
6 \textbf{until} the boundary of $C$ is completely traversed
\end{tabbing}
It is worth to note that, at the beginning of the exploration of the first cell belonging to a tile $S$, 
the quadtree of this tile is set to a single node. However, at the beginning of explorations of subsequent  
cells belonging to $S$, the quadtree of $S$ may be different. So the top-down construction of this quadtree may be spread over the exploration of many cells which will be visited at different points in time.

Consider a tile $T$ and a cell $C\subseteq T$. Let $A_C$ be the area of $C$, $R_C$ be the length of the part of the boundary of $C$ issued from the boundary of $T$, and $P_C$ be the length of the part of the boundary of $C$ issued from the boundary of $\cT$.
\begin{lemma}\label{lem:c}
There is a positive constant $c$, such that $R_C\leq c(A_C/F+P_C)$, for any cell $C$.
\end{lemma}
\begin{proof}
We consider three cases:

\noindent\textbf{Case 1:} $P_C< F/2$ and $A_C< F^2/2$

In this case, we will show that there is a positive constant $c$ such that $R_C\leq c\cdot P_C$.
We call a \emph{borderline} a maximum connected part of the boundary of $\cT$ inside the tile $T$ delimiting the cell $C$. Let $\cL=\{L_1,\ldots,L_l\}$ be the set of borderlines of $C$. There are two types of borderlines: the \emph{linking} borderlines that link two points of the boundary of $T$ and the \emph{closed} borderlines that are closed polygonal lines inside $S$. A borderline $L$ separates the tile $S$ into two connected regions, the \emph{inside} region denoted by $I_L$, i.e., the region containing $C$, and the \emph{outside} region, denoted by $O_L$. If the area of $I_L$ is smaller than that of $O_L$, then $L$ is a \emph{small-inside} borderline, otherwise $L$ is a \emph{large-inside} borderline. We denote by $M_L$ the region among $I_L$ and $O_L$ which has the smaller area.

Notice that the two endpoints of a \emph{linking} borderline can either be on the same side of $S$ or on two adjacent sides. Indeed, any borderline linking two points on opposite sides of $S$ would have a length at least $F$, a contradiction with the inequality $P_C<F/2$. If $L$ is a linking borderline with both endpoints $x$ and $y$ on the same side of $S$, then the length of segment $xy$ is smaller than $|L|$. Hence, the perimeter of $M_L$ is smaller than $2\cdot |L|$. If $L$ is a linking borderline with endpoints $x$ and $y$ on two sides that intersect at a vertex $v$, then the lengths of segments $vx$ and $vy$ are both less than $|L|$. Therefore the perimeter of $M_L$ is smaller than $3\cdot |L|$. If $L$ is a closed borderline, then the perimeter of $M_L$ is exactly $L$. Hence, the perimeter of $M_L$ is always less than $3\cdot |L|$. Moreover, the area of $M_L$ is less than $(3|L|)^2/4\pi$ by the isoperimetric inequality~\cite{Oss78}.

Now we are ready to show that at least one borderline in $\cL$ is a small-inside borderline. 
Suppose, for contradiction that, for all $i=1,\ldots,l$, the borderline $L_i$ is a large-inside borderline. We have $C=S\setminus \bigcup_{i=1}^{l} M_{L_i}$. It follows that:
\begin{eqnarray*}
A_C & = & \area(S)- \sum_{i=1}^{l} \area(M_{L_i})\\
&\geq & F^2-\sum_{i=1}^{l} \frac{(3|L_i|)^2}{4\pi}\quad\mbox{since }(3|L_i|)^2/4\pi\geq \area(M_{L_i})\mbox{ for all }i=1,\ldots,l\\
&\geq & F^2-\frac{9P_C^2}{4\pi}\quad\mbox{since }P_C^2= \left(\sum_{i=1}^{l} |L_i|\right)^2\geq \sum_{i=1}^{l} |L_i|^2\\
&>& \frac{F^2}{2}\quad\mbox{since }F/2> P_C.
\end{eqnarray*}
We obtain $A_C>F^2/2$, a contradiction. This shows that there exists a small-inside borderline
$L\in\cL$. We have $C\subseteq M_{L}$ and thus $R_C< 3 |L|\leq 3P_C$.

\textbf{Case 2:} $P_C< F/2$ and $A_C\geq F^2/2$ 

We have:
$$ R_C  \leq  4F \leq  \frac{8}{F} A_C\quad\mbox{since }F\leq 2A_C/F.$$
\textbf{Case 3:} $P_C\geq F/2$ 

We have:
$$R_C \leq 4F \leq 8P_C\quad\mbox{since }F\leq 2P_C.$$
In all cases, we have $R_C\leq 8(A_C/F+P_C)$.
\end{proof}

The following is the key lemma for all upper bounds proved in this section. Let $\cS=\{T_1,T_2,\dots, T_n\}$ be the set of tiles with non-empty intersection with $\cT$ and  $\cC=\{C_{1},C_{2},\cdots,C_{m}\}$ be the set of cells that are intersections of tiles from $\cS$ with $\cT$. For each $T\in\cS$, let $k_T$ be the number of obstacles of $\cT$ entirely contained in $T$.

\begin{lemma}\label{lem:algo}
For any $F\leq \sqrt{2}/2$, Algorithm $\alg$ explores the terrain $\cT$ of area $A$ and perimeter $P$, using a trajectory of length $O(P+A/F+F\sum_{i=1}^{n}\sqrt{k_{T_i}})$.
\end{lemma} 

\begin{proof}
First, we show that Algorithm $\alg$ explores the terrain $\cT$. Consider the graph $G$ whose vertex set is $\cC$ and edges are the pairs $\{C,C'\}$ such that $C$ and $C'$ have a common point at their boundaries. The graph $G$ is connected, since $\cT$ is connected. Note that for any cell $C$ and
point $r$ on the boundary of $C$, {\tt ExpTrav}$(C,r)$ and thus $\expl(C,r)$
starts and ends on $r$. Therefore, Algorithm $\alg$ performs a depth first traversal of graph $G$, since during the execution of $\expl(C,\dots)$, procedure $\expl(U, \cdots)$ is called for each unvisited cell $U$ adjacent to $C$. Hence, $\expl(C,\dots)$ is called for each cell $C\in\cC$, since $G$ is connected. During the execution of $\expl(C,r)$, $C$ is completely explored by {\tt ExpTrav}($C$,$r$) by the same argument as in the proof of Lemma~\ref{l:explorability}, since the convex hull diameter of $C$ is less than one.

It remains to show that the length of the $\alg$ trajectory is $O(P+A/F+F\sum_{i=1}^{n}\sqrt{k_{T_i}})$. For each $j=1,\ldots,m$, the part of the $\alg$ trajectory inside the cell $C_{j}$ is produced by the execution of $\expl(C_{j},\dots)$. In step 2 of $\expl(C_{j},\dots)$, the robot executes {\tt ExpTrav} with $D=\sqrt{2}F$ and $P=P_{C_{j}}+R_{C_{j}}$. The sum of lengths of recognition and exploration sections of the trajectory in $C_{j}$ is at most $2(P_{C_{j}}+R_{C_{j}})$. The sum of lengths of approaching sections of the trajectory in $T_i$ is at most $6\sqrt{2}F\sqrt{k_{T_i}}$ and each approaching section is traversed twice (cf. proof of Theorem~\ref{t:exploration-correctness}). In step 3 of $\expl(C_{j},\dots)$, the robot only makes the tour of the cell $C_{j}$, hence the distance traveled by the robot is at most $P_{C_{j}}+R_{C_{j}}$. It follows that:
\begin{eqnarray*}
|\alg| & \leq  & 3\sum_{j=1}^{m} (P_{C_{j}}+R_{C_{j}})+12\sqrt{2}F\sum_{i=1}^{n}\sqrt{k_{T_i}}\\
 & \leq & 3\sum_{i=1}^{m}((1+c)P_{C_{j}}+cA_{C_{j}}/F)+12\sqrt{2}F\sum_{i=1}^{n}\sqrt{k_{T_i}}
\quad\mbox{by Lemma~\ref{lem:c}}\\
 & \leq & 3(c+1)P+3cA/F+12\sqrt{2}F\sum_{i=1}^{n}\sqrt{k_{T_i}}.
\end{eqnarray*}
\end{proof}

In view of Lemma~\ref{lem:algo}, exploration of a particular class of terrains can be done at a cost which will be later proved optimal. 
\begin{theorem}\label{th:fat}
Let $c>1$ be any constant. Exploration of a $c$-fat terrain of area $A$, perimeter $P$ and with $k$ obstacles can be performed using a trajectory of length $O(P+A+\sqrt{Ak})$ (without any a priori 
knowledge).
\end{theorem}
\begin{proof}
The robot executes Algorithm $\alg$ with $F=\sqrt{2}/2$. By Lemma~\ref{lem:algo},
the total cost is $O(P+A+\sum_{i=1}^{n}\sqrt{k_{T_i}})$. Recall that $n$ is the number of tiles that have non-empty intersection with the terrain. We have $\sum_{i=1}^{n}\sqrt{k_{T_i}}\leq \sum_{i=1}^{n}\sqrt{\frac{k}{n}}=\sqrt{nk}$. Hence, it remains to show that $n=O(A)$ to prove that
the cost is $O(P+A+\sqrt{Ak})$. By definition of a $c$-fat terrain, there is a disk $D_1$ of radius $r$ included in the terrain and a disk $D_2$ of radius $R$ that contains the terrain, such that $\frac{R}{r}\leq c$. There are $\Theta(r^2)$ tiles entirely included in $D_1$ and hence in the terrain. So, we have $A=\Omega(r^2)$. $\Theta(R^2)$ tiles are sufficient to cover $D_2$ and hence the terrain. So $n=O(R^2)$.
Hence, we obtain $n=O(A)$ in view of $R\leq cr$.
\end{proof}

Consider any terrain $\cT$ of area $A$, perimeter $P$ and with $k$ obstacles. 
We now turn attention to the exploration problem if some knowledge about the terrain is available a priori.
Notice that if $A$ and $k$ are known before the exploration, Lemma~\ref{lem:algo} implies that Algorithm
$\alg$ executed for $F=\min\{\sqrt{A/k},\sqrt{2}/2\}$ explores \emph{any} terrain at cost $O(A+P+\sqrt{Ak})$. (Indeed, if $F=\sqrt{A/k}$ then $A/F=\sqrt{Ak}$ and $kF=\sqrt{Ak}$, while $F=\sqrt{2}/2$ implies $A/F=\Theta(A)$ and $kF=O(A)$.) This cost will be later proved optimal. It turns out that a much more subtle use of Algorithm $\alg$ can guarantee the same complexity assuming only knowledge of $A$ \emph{or} $k$. We present two different algorithms depending on which value, $A$ or $k$, is known to the robot. Both algorithms rely on the same idea. The robot executes Algorithm $\alg$ with some initial value of $F$ until either the terrain is completely explored, or a certain stopping condition, depending on the algorithm, is satisfied. This execution constitutes the first stage of the two algorithms. If exploration was interrupted because of the stopping condition, then the robot proceeds to a new stage by executing Algorithm $\alg$ with a new value of $F$. Values of $F$ decrease in the first algorithm and increase in the second one. The exploration terminates at the stage when the terrain becomes completely explored, while the stopping condition is never satisfied.

In each stage the robot is oblivious of the previous stages, except for the computation of the new value of $F$ that depends on the previous stage. This means that in each stage exploration is done ``from scratch'', without recording what was explored in previous stages. In order to test the stopping condition in a given stage, the robot maintains the following three values: the sum $A^*$ of areas of explored cells, updated after the execution of {\tt ExpTrav} in each cell; the length $P^*$ of the boundary traversed by the robot, continuously updated when the robot moves along a boundary for the first time (i.e., in the recognition mode); and the number $k^*$ of obstacles approached by the robot, updated when an obstacle is approached. The values of $A^*$, $P^*$ and $k^*$ at the end of the $i$-th stage are denoted by $A_i$, $P_i$ and $k_i$, respectively.  Let $F_i$ be the value of $F$ used by Algorithm $\alg$ in the $i$-th stage. Now, we are ready to describe the stopping conditions and the values $F_i$ in both algorithms. 
\begin{center}
\fbox{%
   \begin{minipage}{0.98\textwidth}
\textbf{Algorithm $\algA$, for $A$ known before exploration}

The value of $F$ used in Algorithm $\alg$ for the first stage is $F_1=\sqrt{2}/2$. The value 
of $F$ for subsequent stages is given by $F_{i+1}=\frac{A}{k_iF_i}$. The stopping condition
is $\{k^*F_i\geq 2A/F_i$ and $k^*F_i\geq P^*+1\}$. 
 \end{minipage}%
}

\fbox{%
 \begin{minipage}{0.98\textwidth}
\textbf{Algorithm $\algk$, for $k$ known before exploration}

The value of $F$ used in Algorithm $\alg$ for the first stage is $F_1=\frac{1}{k+\sqrt{2}}$. The value 
of $F$ for subsequent stages is given by $F_{i+1}=\min\left\{\frac{A_i}{kF_i},\frac{\sqrt{2}}{2}\right\}$. The stopping condition is $\{A^*/F_i\geq 2kF_i\mbox{ and }A^*/F_i\geq P^*+1\mbox{ and }F_i<\sqrt{2}/2\}$. 
 \end{minipage}%
}
\end{center}
Consider a moment $t$ during the execution of Algorithm $\alg$. Let $\cC_t$ be the set of cells recorded as visited by Algorithm $\alg$ at moment $t$, and let $\cO_t$ be the set of obstacles approached by the robot until time $t$. For each $C\in \cC_t$, let $B_C$ be the length of the intersection of the exterior boundary of cell $C$ with the boundary of the terrain. For each $O\in \cO_t$, let $|O|$ be the perimeter of obstacle $O$ and let $k_t=|\cO_t|$. The following proposition is proved similarly as Lemma~\ref{lem:algo}.

\begin{proposition}\label{prop:1arg}
There is a positive constant $d$ such that the length of the trajectory of the robot until any time $t$, during the execution of Algorithm $\alg$, is at most $d(\sum_{C\in \cC_t} (B_C+A_C/F)+(k_t+1)\cdot F+\sum_{O\in \cO_t}|O|)$.
\end{proposition}

\begin{proof}
Let $U$ be the current cell at moment $t$, which means that the instance of procedure $\explo$
executed at moment $t$ was called with $U$ as its first parameter. Let $\cO_t'$
be the set of obstacles in $\cO_t$ that are not included in $U$ and let $\cC_t'=\cC_t\setminus \{U\}$.
All the obstacles included in $C\in\cC_t'$ were approached and visited
by the robot. Hence, we have $\sum_{C\in \cC_t'} P_C=\sum_{C\in \cC_t'}B_C +\sum_{O\in \cO_t'}|O|$.
By the same arguments as in the proof of Lemma~\ref{lem:algo}, it follows that the total length 
of the trajectory of the robot in the cells in $\cC_t'$ is at most $d'(\sum_{C\in\cC_t'}(P_C+A_C/F)+F|\cO_t'|)=d'(\sum_{C\in \cC_t'}(B_C+A_C/F)+\sum_{O\in \cO_t'}|O|+F|\cO_t'|)$ for some positive constant $d'$. The length of the trajectory of the robot in $U$ is at most $3G_U+2\sum_{O\in \cO_t\setminus \cO_t'} |O|+|\cO_t\setminus \cO_t'|\cdot F$, where $G_U$ is the length of the exterior boundary of cell $U$. Hence, the length of the trajectory of the robot until time $t$ is at most $d(\sum_{C\in \cC_t} (B_C+A_C/F)+(k_t+1)\cdot F+\sum_{O\in \cO_t}|O|)$ for some positive constant $d$, since $G_U\leq 4F+B_U$. 
\end{proof}

The following lemma establishes the complexity of exploration 
if either the area of the terrain or the number of obstacles is known a priori.

\begin{lemma}\label{lem:algAk}
Algorithm $\algA$ (resp. $\algk$) explores a terrain $\cT$ of area $A$, perimeter $P$ and with $k$ obstacles, using a trajectory of length $O(P+A+\sqrt{Ak})$, if $A$ (resp. $k$) is known before exploration.
\end{lemma} 

\begin{proof}

\textbf{Part 1: complexity of Algorithm $\algA$}

First, we show that the algorithm eventually terminates and completely explores $\cT$. Remark that for each $i> 1$, $F_{i}=\frac{A}{k_{i-1}F_{i-1}}\leq \frac{F_{i-1}}{2}$ since $k_{i-1}F_{i-1}\geq 2\frac{A}{F_{i-1}}$. Since $F_1=\sqrt{2}/2$, for each $i\geq 1$, we have $F_{i}\leq \frac{\sqrt{2}}{2^{i}}$. The algorithm eventually terminates, since there exists $m\in\bbN$ such that $kF_m<2A/F_1<2A/F_m$ and the stopping condition is never satisfied in this case. In the last stage, Algorithm $\alg$ performs complete exploration of the terrain by Lemma~\ref{lem:algo}, since the value of $F$ used for exploration is at most $\sqrt{2}/2$.

Let $D_i$ be the distance traveled by the robot during the $i$-th stage and let $n$ be the number of stages. By Proposition~\ref{prop:1arg}, if stage $i$ ends at moment $t_i$ then $D_i\leq d(\sum_{C\in \cC_{t_i}} (B_C+A_C/F_i)+(|\cO_{t_i}|+1)\cdot F_i+\sum_{O\in \cO_{t_i}}|O|)$ for each $i\geq 1$. Since the algorithm can only interrupt stage $i<n$ when approaching an obstacle, we have $P_i=\sum_{C\in \cC_{t_i}} B_C+\sum_{O\in \cO_{t_i}}|O|$. We obtain that $D_i\leq d(P_i+A/F_i+(k_i+1)F_i)$ for each $i\geq 1$. 

If $n=1$, then the stopping condition is never satisfied and $kF_1\leq \max\{2A/F_1,P+1\}$. The total cost is at most $d(P+A/F_1+(k+1)F_1)=O(P+A)$, since $F_1=\sqrt{2}/2$.
On the other hand, if $n>1$, then for each $i$ such that $1\leq i<n$, we have 
$D_i\leq d(P_i+A/F_i+(k_i+1)F_i)\leq 3d k_iF_i$. Indeed, we have 
$k_iF_i\geq A/F_i$ and $k_iF_i\geq P_i+F_i$, since the stopping 
condition is satisfied at the end of the $i$-th stage. Moreover,
we have $\frac{1}{2}k_{i}F_{i} \geq A/F_{i}=k_{i-1}F_{i-1}$,
for all $1< i<n$. It follows that $\sum_{i=1}^{n-1} D_i\leq 6d k_{n-1}F_{n-1}$.
In order to show that the total cost is $O(P+A+\sqrt{Ak})$, it is sufficient to show that $P+A/F_n+kF_n=O(P+A+\sqrt{Ak})$, since $k_{n-1}F_{n-1}=A/F_n$ and $D_n\leq d(P_n+A/F_n+(k+1)F_n)$.

Take the moment $t_{n-1}$ when the $(n-1)$-th stage was interrupted, i.e., when both inequalities of the stopping condition started to be satisfied. Consider the inequality which was not satisfied just before moment $t_{n-1}$. If this is the first of the two inequalities, then at time $t_{n-1}$ the Algorithm $\algA$
must have increased the value of $k^*$, hence just before moment $t_{n-1}$ we had $(k_{n-1}-1)F_{n-1}<2AF_{n-1}$. Similarly, if the second inequality was not satisfied just before 
moment $t_{n-1}$, then we had $(k_{n-1}-1)F_{n-1}<P_{n-1}+1$.

\begin{description}
\item\textbf{Case 1:} $(k_{n-1}-1)F_{n-1}< 2A/F_{n-1}$

We have:
\begin{eqnarray*}
k_{n-1}F_{n-1}&\leq& 2\frac{A}{F_{n-1}}+F_{n-1}\\
\frac{A}{F_{n}}&\leq& 2\frac{A}{F_{n-1}}+1\quad\mbox{since }F_{n-1}=\frac{A}{k_{n-1}F_{n}} \mbox{ and } F_{n-1}\leq 1\\
\frac{A}{F_{n}}&\leq&  \sqrt{2}\sqrt{Ak_{n-1}}+1\quad\mbox{since }k_{n-1}F_{n-1}\geq 2\frac{A}{F_{n-1}} \mbox{ and thus } \left(\frac{A}{F_{n-1}}\right)^2\leq \frac{1}{2}k_{n-1}F_{n-1}\frac{A}{F_{n-1}}\\
\frac{A}{F_n} &\leq& 2\sqrt{Ak}+1\quad\mbox{since }k_{n-1}\leq k
\end{eqnarray*}

Since the stopping condition in the last stage is not satisfied, we have 
$kF_n\leq\max\{2A/F_n,P+1\}$. We obtain $P+A/F_n+kF_n=O(P+A/F_n)$.
Hence, we have $P+A/F_n+kF_n=O(P+\sqrt{Ak})$.

\item\textbf{Case 2:} $(k_{n-1}-1)F_{n-1}< P_{n-1}+1$

We have:
\begin{eqnarray*}
k_{n-1}F_{n-1} & \leq & P_{n-1}+F_{n-1}+1\\
\frac{A}{F_{n}} & \leq & P_{n-1}+F_{n-1}+1\quad\mbox{since }k_{n-1}F_{n-1}=\frac{A}{F_{n}}\\
\frac{A}{F_{n}} & \leq & P+2\quad\mbox{since }F_{n-1}\leq 1\mbox{ and }P_{n-1}\leq P
\end{eqnarray*}

Since the stopping condition in the last stage is not satisfied, we have 
$kF_n\leq\max\{2A/F_n,P+1\}$, as before.
We obtain $P+A/F_n+kF_n=O(P+A/F_n)=O(P)$.
\end{description}

\noindent\textbf{Part 2: complexity of Algorithm $\algk$}

The proof is similar to that concerning Algorithm $\algA$. The main difference follows from the additional clause $F_i<\sqrt{2}/2$ in the stopping condition. This clause was not necessary in Algorithm $\algA$
because, as opposed to the present case, sides of tiles were decreasing. 
First, we show that the algorithm eventually terminates and completely explores $\cT$. Remark that for each $i>1$, we have $F_{i}=\min\{A_{i-1}/kF_{i-1},\sqrt{2}/2\}\geq\min\{2F_{i-1},\sqrt{2}/2\}$, since $A_{i-1}/F_{i-1}\geq 2kF_{i-1}$. Hence, the algorithm eventually terminates. Indeed, even if the first two inequalities remain true, the third must
become false at some point. Notice that $F_1\leq \sqrt{2}/2$ since $k\geq 0$, and for all $i> 1$ we have $F_{i}\leq \sqrt{2}/2$. In the last stage, Algorithm $\alg$ performs complete exploration of the terrain by Lemma~\ref{lem:algo}, since the value of $F$ used for the exploration is at most $\sqrt{2}/2$.

Let $D_i$ be the distance traveled by the robot during the $i$-th stage and let $n$ be the number of stages. By Proposition~\ref{prop:1arg}, if stage $i$ ends at moment $t_i$, then $D_i\leq d(\sum_{C\in \cC_{t_i}} (B_C+A_C/F_i)+(|\cO_{t_i}|+1)\cdot F_i+\sum_{O\in \cO_{t_i}}|O|)$ for each $i\geq 1$. Since the algorithm can only stop when completing the exploration of a cell, we have $P_i=\sum_{C\in \cC_{t_i}} B_C+\sum_{O\in \cO_{t_i}}|O|$ and $A_i=\sum_{C\in \cC_{t_i}} A_C$. We obtain that $D_i\leq d(P_i+A_i/F_i+(k+1)F_i)$ for each $i\geq 1$. 

If $n=1$, then the stopping condition is never satisfied and either $A/F_1\leq \max\{2kF_1,P+1\}$ or $F_1=\sqrt{2}/2$. In the first case, the total cost is at most $d(P+A/F_1+(k+1)F_1)=O(P)$ since $kF_1\leq 1$. In the second case, we have $k=0$ and  the total cost is at most $d(P+A/F_1+(k+1)F_1)=O(P+A)$.
On the other hand, if $n>1$ then for each $i$ such that $1\leq i<n$, we have 
$D_i\leq d(P_i+A_i/F_i+(k+1)F_i)\leq 3d A_i/F_i$. Indeed, we have 
$\frac{A_i}{F_i}\geq kF_i$ and $\frac{A_i}{F_i}\geq P_i+F_i$, since the stopping 
condition is satisfied at the end of the $i$-th stage. Moreover,
we have $\frac{A_{i}}{2F_{i}} \geq kF_{i}=\frac{A_{i-1}}{F_{i-1}}$
for all $1< i<n$. It follows that $\sum_{i=1}^{n-1} D_i\leq 6d \frac{A_{n-1}}{F_{n-1}}$.

Notice that the third inequality of the stopping condition is always satisfied during 
the $(n-1)$-stage. Take the moment $t_{n-1}$ when the $(n-1)$-th stage was interrupted, i.e., when the two first inequalities of the stopping condition started to be satisfied. Consider the inequality which was not satisfied just before moment $t_{n-1}$. If this is the first of the two inequalities, then at time $t_{n-1}$ the Algorithm $\algk$ must have increased the value of $A^*$ by at most $F_{n-1}^2$ since each cell
is included in a square of size $F_{n-1}$. Hence just before moment $t_{n-1}$ we had $
\frac{A_{n-1}-F_{n-1}^2}{F_{n-1}}<2AF_{n-1}$. Similarly, if the second inequality was not satisfied just before moment $t_{n-1}$, then we had $\frac{A_{n-1}-F_{n-1}^2}{F_{n-1}}<P_{n-1}+1$.

\begin{description}
\item\textbf{Case 1:} $\frac{A_{n-1}-F_{n-1}^2}{F_{n-1}}<2AF_{n-1}$

Notice that $k\geq 1$, since otherwise $F_1=\sqrt{2}/2$ and the 
stopping condition would never be satisfied in the first stage.
We have:
\begin{eqnarray*}
\frac{A_{n-1}}{F_{n-1}} & \leq & 2kF_{n-1}+F_{n-1}\\
kF_n &\leq &  2kF_{n-1}+1\quad\mbox{since }
F_{n-1}\leq\frac{A_{n-1}}{kF_{n}}\mbox{ and }F_{n-1}\leq 1\\
kF_n &\leq & \sqrt{2}\sqrt{A_{n-1}k}+1\quad\mbox{since }\frac{A_{n-1}}{F_{n-1}}\geq 2kF_{n-1} \mbox{ and thus } \left(kF_{n-1}\right)^2 \leq \frac{1}{2}\frac{A_{n-1}}{F_{n-1}}kF_{n-1}\\
kF_n &\leq & 2\sqrt{Ak}+1\quad\mbox{since }A_{n-1}\leq A
\end{eqnarray*}

Since the stopping condition in the last stage is not satisfied, we have 
either $A/F_n\leq\max\{2kF_n,P+1\}$ and thus $A/F_n=O(P+kF_n)$,
or $F_n=\sqrt{2}/2$ and thus $A/F_n=O(A)$. We obtain that $D_n\leq d(P+A/F_n+(k+1)F_n)=O(P+\sqrt{Ak}+A)$. From the previous sequence of inequalities, we also have $\sum_{i=1}^{n-1} D_i\leq 6d(A_{n-1}/F_{n-1}) =O(\sqrt{Ak})$. Hence, the total cost is $O(P+\sqrt{Ak}+A)$.

\item\textbf{Case 2:} $\frac{A_{n-1}-F_{n-1}^2}{F_{n-1}}<P_{n-1}+1$

We have:
\begin{eqnarray*}
\frac{A_{n-1}}{F_{n-1}} & \leq & P_{n-1}+F_{n-1}+1\\
kF_{n} & \leq & P_{n-1}+F_{n-1}+1\quad\mbox{since }kF_{n}\leq\frac{A_{n-1}}{F_{n-1}}\\
kF_{n} & \leq & P+2\quad\mbox{since }F_{n-1}\leq 1\mbox{ and }  P_{n-1}\leq P
\end{eqnarray*}

As before, we have $A/F_n=O(P+kF_n)$ or $A/F_n=O(A)$.
We obtain that $D_n\leq d(P+A/F_n+(k+1)F_n)=O(P+A)$. 
We have $\sum_{i=1}^{n-1} D_i\leq 6d(A_{n-1}/F_{n-1}) =O(P)$. Hence, 
the total cost is $O(P+A)$.

\end{description}
\end{proof}

The following theorem shows that the lengths of trajectories in Lemma~\ref{lem:algAk} and in Theorem~\ref{th:fat} are asymptotically optimal.

\begin{theorem}
\label{t:lower-limited}
Any algorithm for a robot with limited visibility, exploring polygonal terrains of area $A$, perimeter $P$ and with $k$ obstacles, produces trajectories of length $\Omega(P+A+\sqrt{Ak})$ in some terrains, even if the terrain is known to the robot.
\end{theorem}

\begin{proof}
In order to prove our lower bound we present three families of terrains. A terrain in the first family 
(cf. Fig. \ref{fig:lower-unlimited}(a)) is a square with identical obstacles of diameter $\epsilon$ located on a $\sqrt{k}\times\sqrt{k}$ grid. The side of the square is $\sqrt{A+x}$, where $x$ is the negligible total area of all $k$ obstacles and the total perimeter of all obstacles is $1$. By the same arguments as in the proof of Theorem~\ref{t:lower-unlimited}, any exploration trajectory must be of length at least $(k-1)\left(\frac{\sqrt{A+x}}{\sqrt{k}+1}-\epsilon\right)$, which is in $\Omega(\sqrt{Ak})$. At the same time we have $P= \Theta(\sqrt{A})$ ($P$ cannot be smaller). A terrain in the second family (cf. Fig. \ref{fig:lower-unlimited}(b)) is a polygon of arbitrarily small area (without obstacles), whose exploration requires a trajectory of size $\Omega(P)$. A terrain in the third family is the empty square of side $\sqrt{A}$. Now we have $P=\Theta(\sqrt{A})$ and $k=0$. When the robot traverses a distance $d$, it explores a new area of at most $2d$. So, the robot has to travel a distance $\Omega(A)$ to explore such a terrain. These three families of terrains justify our lower bound.
\end{proof}

The examples from the above proof can be adjusted so that all considered terrains are $c$-fat,
for any fixed constant $c>1$. Lemma~\ref{lem:algAk} and Theorem~\ref{t:lower-limited} imply

\begin{theorem}\label{th:main}
Consider terrains of area $A$, perimeter $P$ and with $k$ obstacles.
If either $A$ or $k$ is known before the exploration, then the exploration of any such terrain can be performed using a trajectory of length $\Theta(P+A+\sqrt{Ak})$, which is asymptotically optimal.
\end{theorem}

Notice that in order to explore a terrain at cost $O(P+A+\sqrt{Ak})$, it is enough to know the parameter $A$ or $k$ up to a multiplicative constant, rather than the exact value. This can be proved by a carefull modification of the proof of Lemma~\ref{lem:algAk}. For the sake of clarity, we stated and proved the weaker version of Lemma~\ref{lem:algAk}, with knowledge of the exact value. 
 
Suppose now that no a priori knowledge of any parameters of the terrain is available. We iterate Algorithm $\algA$ or $\algk$ for $A$ (resp. k) equal $1,2,4,8,\ldots$ interrupting the iteration and doubling the 
parameter as soon as the explored area (resp. the number of obstacles seen) exceeds the current parameter value. The algorithm stops when the entire terrain is explored (which happens at the first probe exceeding the actual unknown value of $A$, resp. $k$). We get an exploration algorithm using a trajectory of length $O((P+A+\sqrt{Ak})\log A)$, resp. $O((P+A+\sqrt{Ak})\log k)$. By interleaving the two procedures we get the minimum of the two costs. Thus we have the following corollary. 

\begin{corollary}
Consider terrains of area $A$, perimeter $P$ and with $k$ obstacles.
Exploration of any such terrain can be performed without any a priori knowledge at cost differing from the worst-case optimal cost with full knowledge only by a factor $O(\min\{\log A,\log k\})$.
\end{corollary}


\begin{thebibliography}{12}

\bibitem{AKS02}
S. Albers and K. Kursawe and S. Schuierer,
Exploring unknown environments with obstacles,
Algorithmica 32 (2002), 123--143.

\bibitem{BLAS}
T. Bandyopadhyay, Z. Liu, M.H Ang, M.H, W.K.G Seah, 
Visibility-based exploration in unknown environment containing structured obstacles, 
Advanced Robotics (2005), 484-491.

\bibitem{BBFY}
E. Bar-Eli, P. Berman, A. Fiat and R. Yan,
On-line navigation in a room,
Journal of Algorithms 17 (1994), 319-341.

\bibitem{BPFKRS}
P. Berman, A. Blum, A. Fiat, H. Karloff, A. Rosen and M. Saks,
Randomized robot navigation algorithms,
Proc. 7th ACM-SIAM Symp. on Discrete Algorithms (1996), 74-84.

\bibitem{BRS}
A. Blum, P. Raghavan and B. Schieber,
Navigating in unfamiliar geometric terrain,
SIAM Journal on Computing 26 (1997), 110-137.

\bibitem{CQG}
N. Cuperlier, M. Quoy, C. Giovanangelli, 
Navigation and planning in an unknown environment using vision and a cognitive map, 
Proc. Workshop: Reasoning with Uncertainty in Robotics (2005), 48-53.

\bibitem{DKP91}
X. Deng, T. Kameda and C. H. Papadimitriou,
How to learn an unknown environment,
Proc. 32nd Symp. on Foundations of Computer Science (FOCS 1991),
298--303.

\bibitem{DKP98}
X. Deng, T. Kameda and C. H. Papadimitriou,
How to learn an unknown environment I: the rectilinear case,
Journal of the ACM 45 (1998), 215-245.

\bibitem{GB01}
Y. Gabriely, E. Rimon,
Spanning-tree based coverage of continuous areas by a mobile robot,
Proc. Int. Conf. of Robotics and Automaton (ICRA 2001), 1927-1933.

\bibitem{GB03}
Y. Gabriely and E. Rimon,
Competitive on-line coverage of grid environments by a mobile robot,
Computational Geometry: Theory and Applications (2003), 24(3):197-224.

\bibitem{GBBS08}
S.K. Ghosh, J.W. Burdick, A. Bhattacharya and S. Sarkar,
Online algorithms with discrete visibility - exploring unknown polygonal environments,
Robotics \& Automation Magazine 15 (2008), 67-76.

\bibitem{HIKK01}
F. Hoffmann. C. Icking, R. Klein and K. Kriegel,
The polygon exploration problem,
SIAM J. Comput. 31 (2001), 577--600.

\bibitem{IKRL00}
C. Icking, T. Kamphans, R. Klein and E. Langetepe. Exploring an unknown cellular environment.
In Abstracts of the 16th European Workshop on Computational Geometry, pages 140-143, 2000.

\bibitem{KKMZ09}
A. Kolenderska, A. Kosowski, M. Ma{\l}afiejski and P. \.Zyli\'nski. 
An Improved Strategy for Exploring a Grid Polygon, SIROCCO 2009.

\bibitem{N92}
S. Ntafos, Watchman routes under limited visibility,
Comput. Geom. Theory Appl. 1 (1992), 149--170.

\bibitem{Oss78}
R.~Osserman.
\newblock The isoperimetric inequality.
\newblock {\em Bull. Amer. Math. Soc.}, 84:1182--1238, 1978.

\bibitem{PY}
C. H. Papadimitriou, M. Yannakakis,
Shortest paths without a map,
Theor. Comput. Sci. 84 (1991), 127--150.

\bibitem{SM}
R. Sim, J.J. Little, 
Autonomous vision-based exploration and mapping using hybrid maps and Rao-Blackwellised particle filters,
Intelligent Robots and Systems (2006), 2082-2089.

\bibitem{TML}
B. Tovar, R. Murrieta-Cid, S. M. Lavalle, 
Distance-optimal navigation in an
unknown environment without sensing distances, 
IEEE Transactions on Robotics 23 (2007), 506-518.  
\end{thebibliography}
\end{document}